\documentclass[journal]{IEEEtran}
\usepackage{amsmath,amssymb,mathtools,bm}
\usepackage{graphicx}
\usepackage{cite}
\usepackage{balance}
\usepackage{url}
\graphicspath{{simulations/figures/}}

\newtheorem{assumption}{Assumption}
\newtheorem{definition}{Definition}
\newtheorem{lemma}{Lemma}
\newtheorem{proposition}{Proposition}
\newtheorem{theorem}{Theorem}
\newtheorem{corollary}{Corollary}
\newtheorem{remark}{Remark}

\newcommand{\R}{\mathbb{R}}
\newcommand{\T}{^{\top}}
\newcommand{\norm}[1]{\lVert #1\rVert}
\newcommand{\diag}{\operatorname{diag}}
\newcommand{\He}{\operatorname{He}}

\begin{document}

\title{Resilient Control of Switched Vehicle Platoons under False Data Injection Attacks}

\author{Ali~Eslami and Jiangbo~Yu%
\thanks{Ali Eslami (ali.eslami@mcgill.ca) and Jiangbo Yu (jiangbo.yu@mcgill.ca) are with the Department of Civil Engineering,
McGill University, Montreal, QC, H3A 0C3, Canada. \textit{(Corresponding author: Jiangbo Yu.)}}}

\maketitle

\begin{abstract}
This paper investigates resilient control design for leader--follower vehicle platoons with mode-dependent powertrain dynamics subject to False Data Injection (FDI) attacks on vehicle-to-vehicle (V2V) communication channels. The longitudinal motion of each vehicle is described by a switched third-order model that captures changes in the powertrain dynamics across different operating modes. To estimate the attack signals that are injected into the communication channels, each vehicle is equipped with an auxiliary system, and a dedicated observer is implemented for each communication link. The resulting attack estimates are then used to mitigate the effects of the attacks through our proposed resilient controller. For attacks with bounded rates but potentially unbounded amplitudes, the closed-loop platoon is shown to be uniformly ultimately bounded. For a predecessor-following topology, string stability is established for the nominal switched platoon, while the effect of nonzero attack-estimation errors on acceleration propagation is shown to be bounded. Numerical case studies demonstrate the effectiveness of the proposed approach.
\end{abstract}


\section{Introduction}
\label{sec:introduction}

\IEEEPARstart{C}{ooperative} vehicle platooning coordinates a group of connected
vehicles to maintain desired intervehicle spacing while possibly tracking the leader's
motion \cite{Rajamani12,Zheng16,yu2025agentic}. Such coordinated operation can improve road
capacity and traffic efficiency while reducing energy consumption and
emissions \cite{Petrillo18,Mousavinejad22}. To achieve this coordination,
vehicles exchange motion information, such as position, velocity, and
acceleration, through vehicle-to-vehicle (V2V) communication links
\cite{Zheng16,Mousavinejad22}. This reliance on wireless communication,
however, exposes the platoon to cyberattacks that can corrupt the exchanged
information, degrade platoon performance, and potentially compromise safe
operation \cite{Petrillo18,Mousavinejad22,eslami2025security,Ploeg14,eslami2026stable}.

From a control perspective, two fundamental requirements are internal
stability, which concerns convergence of the closed-loop platoon states for a
fixed platoon size, and string stability, which prevents disturbances from
being amplified as they propagate downstream along the vehicle chain
\cite{Swaroop96,Ploeg14,Zheng16}. For predecessor-following architectures, a
constant-time-headway (CTH) spacing policy is widely used to ensure string
stability when the headway and controller parameters satisfy appropriate
conditions \cite{Ioannou93,Naus10}. Classical CTH and string-stability
analyses, however, predominantly consider vehicles with fixed longitudinal
dynamics.

In practical operation, a vehicle's longitudinal dynamics may vary with gear
selection, torque source, propulsion mode, or actuation regime. These
operating-regime changes can be represented by mode-dependent powertrain
dynamics, resulting in a vehicle model that switches among distinct
longitudinal modes. Stability of every frozen mode does not, in general,
guarantee stability of the resulting switched closed loop
\cite{Hespanha99,Liberzon03}. Average-dwell-time and
multiple-Lyapunov-function techniques provide standard tools for analyzing
such systems when the modes do not share a common Lyapunov function
\cite{Hespanha99,Liberzon03}.

Related platooning studies have considered switching primarily at the
communication layer. Wang et al.\ developed a model-predictive controller for
reconfigurable communication topologies, whereas Ding et al.\ considered
distributed adaptive platooning under Markov switching graphs
\cite{wang2022mpc,ding2024distributed}. In contrast, the switching considered
here occurs in the physical vehicle dynamics through mode-dependent powertrain
lags.

Separately, the dependence of cooperative platooning on V2V communication
creates vulnerability to false-data-injection (FDI) attacks. The construction
and effects of attacks on cooperative driving have been studied in
\cite{eslami2025resource,Teixeira15}. The authors in \cite{Biron18} developed real-time
denial-of-service attack detection and estimation, while \cite{Mousavinejad20}
proposed distributed attack detection and recovery. Secure adaptive and observer-based control under
communication delays and FDI attacks was investigated in
\cite{Petrillo21,Wang25}. The authors in \cite{Guo24} designed a distributed finite-time
observer for joint vehicle-state and input-channel attack estimation, whereas
\cite{Li26,Song26} considered switching-topology and
hybrid-stochastic attack models, respectively. More
recently, \cite{gong2026toward} proposed a two-layer architecture against coupled FDI
and denial-of-service attacks, actuator faults, and external disturbances. Redundancy-based resilient-control and secure-estimation
methods, however, require sufficiently many uncompromised information sources
or impose attack-sparsity conditions \cite{Safeguard24,Wei24}.

These research directions reveal an important gap. Existing observer-based
and resilient-control methods generally consider fixed physical dynamics and
prescribed, stochastic, or finite-energy attack models, while classical CTH
analyses predominantly assume fixed, attack-free vehicle dynamics. The joint
problem of guaranteeing platoon stability under switching powertrain dynamics
and bounded-rate FDI attacks with potentially unbounded amplitudes, when every
controller-relevant V2V link may be corrupted, therefore remains
insufficiently addressed.

This paper addresses these physical and cyber challenges simultaneously. Each
vehicle is modeled as a switched third-order system whose effective powertrain
lag changes across operating modes, and every V2V communication channel may be
subject to bounded-rate FDI attacks with possibly unbounded amplitudes. To
estimate and compensate for the corrupted information, each sender runs a
defender-designed switching auxiliary system, and each receiver implements a
dedicated observer for every communication link. The resulting design renders
the compensated platoon uniformly ultimately bounded and, for a
predecessor-following topology, nominally switched string stable. It also
provides a certified bound on acceleration propagation under nonzero
attack-estimation errors.

Therefore, the main contributions of this paper are summarized as follows:
\begin{enumerate}
\item An augmented observer is developed to estimate the FDI signals added
to the communicated vehicle states and auxiliary outputs on each
communication link. The estimation error depends on the rate of change of
the attacks rather than their amplitudes. The method does not require a
bound on the attack amplitudes or on the number of attacked links.

\item A resilient control protocol is proposed to achieve uniformly ultimately
bounded results for the platoon subject to bounded-rate FDI attacks with potentially unbounded amplitudes.

\item For the predecessor-following topology, string stability is established for the nominal switched platoon. A bound is also derived for acceleration propagation when the attack-estimation errors are nonzero, under potentially unbounded attack amplitudes with bounded rates.

\end{enumerate}
The remainder of the paper is organized as follows.
Section~\ref{sec:formulation} formulates the communication graph, switched
vehicle dynamics, and attack and auxiliary channels.
Section~\ref{sec:estimation} establishes augmented-system observability and
derives an amplitude-independent attack-estimation bound. Section~\ref{sec:stability} analyzes the stability of the platoon. Section~\ref{sec:sim} presents the simulation case scenarios, and Section~\ref{sec:conclusion} concludes the paper.

\emph{Notation:} $\norm{\cdot}$ denotes the Euclidean norm for vectors and the
corresponding induced norm for matrices. For a signal $f$ and any $T>0$, its
finite-horizon $\mathcal{L}_2$ norm is defined by
\begin{align*}
\norm{f}_{2,T}
:=
\left(\int_{0}^{T}\norm{f(t)}^{2}\,dt\right)^{1/2}
\end{align*}
The power seminorm is defined as
\begin{equation}
\label{eq:powernorm}
\norm{f}_{\mathcal{P}}
:=
\left(
\limsup_{T\to\infty}
\frac{1}{T}
\int_{0}^{T}\norm{f(t)}^{2}\,dt
\right)^{1/2}
\end{equation}
The time argument of a time-varying signal is omitted when it is clear from
the context. It is shown explicitly when signals are evaluated at different
time instants or when their temporal dependence must be emphasized.

\section{Problem Formulation}\label{sec:formulation}
This section specifies the communication graph, the switched vehicle dynamics, the defender-designed auxiliary system, and the cyber-attack signals on the communication channels.
\subsection{Communication Graph}
There is one leader, indexed by $0$, and $N$ followers. Let
$M=[m_{ij}]$ be the follower adjacency matrix, where $m_{ij}=1$ when follower
$i$ receives vehicle $j$'s packet. Let $L$ be its Laplacian matrix and
$P=\diag(p_1,\ldots,p_N)$, where $p_i=1$ if follower $i$ receives information directly from the leader.
Define $\mathcal N_i:=\{j:m_{ij}=1\}$ and let
$\mathcal P_i=\{0\}$ if $p_i=1$ and $\mathcal P_i=\emptyset$ otherwise. Then
\begin{equation*}
 \mathcal I_i:=\mathcal N_i\cup\mathcal P_i,\qquad
 \mathcal E_c:=\{(i,j):j\in\mathcal I_i\}.
\end{equation*}

The communication topology is required to satisfy the following connectivity
condition.
\begin{assumption}\label{as:graph}
The graph is fixed and look-ahead, and its leader-augmented graph contains a
spanning tree rooted at vehicle $0$. Thus $L+P$ is nonsingular and, for the
look-ahead ordering, has real positive eigenvalues
$\lambda_1,\ldots,\lambda_N$.
\end{assumption}

\subsection{Switched Vehicle Dynamics}
\label{subsec:Switched Vehicle dynamics}
For each vehicle $i\in\{0,\ldots,N\}$, consider the following switched
dynamics:
\begin{equation}\label{eq:vehicle}
 \begin{aligned}
 \dot x_i(t)&=A_{\sigma(t)}x_i(t)+B_{\sigma(t)}u_i(t),\\
 x_i(t)&=\begin{bmatrix}s_i(t)&v_i(t)&a_i(t)\end{bmatrix}\T.
 \end{aligned}
\end{equation}
Here, $\sigma(t)$ denotes the active physical mode at time $t$. Define the
physical mode set $\mathcal S:=\{1,\ldots,n_s\}$. For each arbitrary mode
$p\in\mathcal S$, let $\tau_p>0$ be the mode-dependent powertrain lag, and define
\begin{equation*}
 A_p=\begin{bmatrix}0&1&0\\0&0&1\\0&0&-1/\tau_p\end{bmatrix},
 \qquad B_p=\begin{bmatrix}0\\0\\1/\tau_p\end{bmatrix}.
\end{equation*}
The physical switching signal $\sigma:[0,\infty)\to\mathcal S$ is
right-continuous, piecewise constant, and common to the fleet; thus
$\sigma(t)$ denotes the active physical mode at time $t$, whereas $p$ denotes
an arbitrary element of $\mathcal S$.

\begin{remark}
The switched model captures changes in the effective longitudinal response
caused by gear shifts, propulsion modes, or actuation regimes. It can be
obtained from a nonlinear force-balance model after compensation of the
nominal resistance terms, yielding
$\tau_p\dot a_i+a_i=u_i$ in mode $p$ \cite{eslami2025resource}.
\end{remark}

\begin{remark}
As a first step, this work assumes fleet-synchronous physical switching and instantaneous synchronization between the active powertrain mode and the corresponding controller gain. A natural extension is to consider asynchronous switching, in which the plant and controller switching signals may temporarily differ. In that setting, suitable multiple Lyapunov functions can be used to characterize state contraction during matched intervals and possible growth during mismatched intervals. By imposing an upper bound on the mismatch duration or mismatch-time ratio, together with appropriate dwell-time conditions, one may ensure that the decay accumulated during matched operation dominates the mismatch-induced growth. Extending the proposed resilient-control and string-stability analysis to this asynchronous setting is left for future work.
\end{remark}

\begin{assumption}\label{as:physical-switching}
The physical switching signal has average dwell time (ADT) $\tau_a>0$ and
chatter bound $N_0\ge0$, namely
\begin{equation}\label{eq:adt}
 N_\sigma(t,s)\le N_0+\frac{t-s}{\tau_a},\qquad t\ge s\ge0.
\end{equation}
\end{assumption}

Introduce the common moving reference
\begin{equation}\label{eq:common-frame}
 \begin{aligned}
 r_c(t)&=\begin{bmatrix}s_c+v_ct&v_c&0\end{bmatrix}\T,\\
 \dot r_c(t)&=\begin{bmatrix}v_c&0&0\end{bmatrix}\T,
 \end{aligned}
\end{equation}
where $s_c$ and $v_c$ are pre-agreed protocol parameters known to all
legitimate vehicles. It follows that for every
$p\in\mathcal S$,
\begin{equation*}
 A_pr_c(t)
 =
 \begin{bmatrix}v_c&0&0\end{bmatrix}\T
 =\dot r_c(t).
\end{equation*}
Consequently, the active physical mode also satisfies
$A_{\sigma(t)}r_c(t)=\dot r_c(t)$. Define the moving-frame state
$\bar x_j(t):=x_j(t)-r_c(t)$. Its dynamics follow explicitly as
\begin{equation}\label{eq:moving-frame-dynamics}
 \begin{aligned}
 \dot{\bar x}_j(t)
 &=\dot x_j(t)-\dot r_c(t)\\
 &=A_{\sigma(t)}x_j(t)+B_{\sigma(t)}u_j(t)
   -A_{\sigma(t)}r_c(t)\\
 &=A_{\sigma(t)}\bigl(x_j(t)-r_c(t)\bigr)
   +B_{\sigma(t)}u_j(t)\\
 &=A_{\sigma(t)}\bar x_j(t)+B_{\sigma(t)}u_j(t).
 \end{aligned}
\end{equation}
Moreover,
\begin{align}\label{eq:moving-frame-bound}
 \bar x_j(t)&=
 \begin{bmatrix}s_j(t)-s_c-v_ct&v_j(t)-v_c&a_j(t)\end{bmatrix}\T,\nonumber\\
 s_j(t)-s_c-v_ct
 &=s_j(t_0)-s_c-v_ct_0+\int_{t_0}^{t}(v_j(\tau)-v_c)\,d\tau.
\end{align}
The arguments $t$, $t_0$, and $\tau$ are displayed in
\eqref{eq:moving-frame-bound} for temporal clarity.
Thus $\bar x_j$ is bounded if
$v_j-v_c\in\mathcal L_1\cap\mathcal L_\infty$ and $a_j$ is bounded. This fact will be used in the proposed auxiliary system presented in the next section.

\subsection{Defender-Designed Auxiliary System}
Each sender $j$ runs a switching auxiliary system with state $z_j\in\R^{n_z}$ and
output $y_j^z\in\R^{m_z}$:
\begin{equation}\label{eq:auxiliary}
 \dot z_j(t)=A_{z,\rho(t)}z_j(t)+B_{z,\rho(t)}\bar x_j(t),
 \qquad y_j^z(t)=C_{z,\rho(t)}z_j(t).
\end{equation}
Here, $\rho(t)$ denotes the active auxiliary mode at time $t$. Define the
auxiliary mode set $\mathcal R:=\{1,\ldots,n_r\}$ and the defender-designed
auxiliary switching signal $\rho:[0,\infty)\to\mathcal R$, which is
independent of the physical switching signal $\sigma$ in (\ref{eq:vehicle}). Thus, $r$ denotes an
arbitrary element of $\mathcal R$, whereas $\rho(t)$ denotes the active
element. For a family of mode-indexed matrices
$\{M_r\}_{r\in\mathcal R}$, $M_{\rho(t)}$ denotes the matrix associated with
the active auxiliary mode. When the time dependence is clear, we use the
shorthand $M_\rho:=M_{\rho(t)}$.

For each $r\in\mathcal R$, the defender designs
$A_{z,r}\in\R^{n_z\times n_z}$,
$B_{z,r}\in\R^{n_z\times3}$, and
$C_{z,r}\in\R^{m_z\times n_z}$, as well as the switching schedule $\rho$.
The schedule uses a prescribed mode cycle $(r_1,\ldots,r_m)$, where
$r_\ell\in\mathcal R$ and $m$ is the number of mode dwells in one cycle. For
each $\ell\in\{1,\ldots,m\}$, the defender specifies a finite dwell-time set
\[
 \mathcal D_\ell\subset
 [\underline\Delta_\ell,\overline\Delta_\ell],
 \qquad
 0<\underline\Delta_\ell\le\overline\Delta_\ell<\infty.
\]
In cycle $k$, $\Delta_{\ell,k}\in\mathcal D_\ell$ denotes the dwell time of
mode $r_\ell$. The maximum cycle duration is
\begin{equation}\label{eq:cycle-time}
 T_\rho:=\sum_{\ell=1}^m\overline\Delta_\ell.
\end{equation}
These defender-selected constants specify the admissible auxiliary schedules
and are used below to select a uniform observability window and quantify the
observer decay rate.

\subsection{Communication and Attack Model}
For each communication link $(i,j)\in\mathcal E_c$, vehicle $j$ transmits its
physical state in the original coordinates and its auxiliary output. Receiver
$i$ obtains the corrupted physical-state packet
\begin{equation}\label{eq:attacked-packets}
 x_{ij}^*=x_j+x^a_{ij},
\end{equation}
and the corrupted auxiliary-output packet
\begin{equation}\label{eq:attacked-aux-packet}
 y_{ij}^{z*}=C_{z,\rho(t)}z_j+y^a_{ij}.
\end{equation}
Here, $x^a_{ij}\in\R^3$ is the additive attack on the physical-state packet,
and $y^a_{ij}\in\R^{m_z}$ is the additive attack on the auxiliary-output
packet. Therefore, the adversary may corrupt both the
ordinary physical-state channel and the auxiliary-output channel. Since the physical state is transmitted in the original
coordinates, receiver $i$ constructs the moving-frame input
\begin{equation}\label{eq:normalized-packet}
 \bar x_{ij}^*:=x_{ij}^*-r_c(t)=\bar x_j+x^a_{ij}
\end{equation}
Note that as stated in Section \ref{subsec:Switched Vehicle dynamics}, $r_c$ is known to all the vehicles.

The admissible attack signals are assumed to satisfy the following conditions.
\begin{assumption}\label{as:attack}
For every $(i,j)\in\mathcal E_c$, $x^a_{ij}$ and $y^a_{ij}$ have finite
initial values and are locally absolutely continuous. Furthermore, the attack rates satisfy for an unknown finite
$\bar d_{ij}$, the following condition:
\begin{equation}\label{eq:attack-rate}
 \left\|\begin{bmatrix}\dot y^a_{ij}(t)\\\dot x^a_{ij}(t)\end{bmatrix}\right\|
 \le\bar d_{ij}\quad\text{for almost every }t.
\end{equation}
\begin{remark}
    Note that in Assumption \ref{as:attack}, no bounds are imposed on the attack-signal amplitudes.
\end{remark}
\end{assumption}

\begin{assumption}\label{as:aux-schedule}
The sender and receiver use the same switching signal $\rho$. It repeats the prescribed
mode cycle $(r_1,\ldots,r_m)$, and in every cycle $k$, mode $r_\ell$ is active
for a dwell time $\Delta_{\ell,k}$ selected from the corresponding pre-agreed
set $\mathcal D_\ell$.
\end{assumption}

\begin{remark}
A pre-shared seed and synchronized clock allow the sender and receiver to
reproduce $\rho$, so no secure online channel is required.
\end{remark}

\section{Attack Estimation}\label{sec:estimation}
Consider an arbitrary communication link $(i,j)\in\mathcal E_c$. For
notational simplicity, the link indices $i$ and $j$ are omitted throughout
the remainder of this section. Define
\begin{equation}\label{eq:aug-state}
 \chi=\begin{bmatrix}z_j\\y^a\\x^a\end{bmatrix},\quad
 \nu=\begin{bmatrix}\dot y^a\\\dot x^a\end{bmatrix},
\end{equation}
and, for $r\in\mathcal R$,
\begin{align}\label{eq:aug-matrices}
 \mathcal A_r&=
 \begin{bmatrix}A_{z,r}&0&-B_{z,r}\\0&0&0\\0&0&0\end{bmatrix},
 &\mathcal B_r&=\begin{bmatrix}B_{z,r}\\0\\0\end{bmatrix},\nonumber\\
 \mathcal C_r&=\begin{bmatrix}C_{z,r}&I&0\end{bmatrix},
 &D&=\begin{bmatrix}0&0\\I&0\\0&I\end{bmatrix}.
\end{align}
Therefore, we have the following augmented dynamics:
\begin{equation}\label{eq:aug-system}
 \dot\chi=\mathcal A_\rho\chi+\mathcal B_\rho\bar x^*+D\nu,
 \qquad y^{z*}=\mathcal C_\rho\chi.
\end{equation}

Let $\Phi_{\mathcal A}(t,s)$ be the transition matrix of the homogeneous
system generated by $\mathcal A_{\rho(t)}$.

We first define the uniform observability property required for reconstructing
the augmented state under auxiliary switching.
\begin{definition}\label{def:uco}
The pair $(\mathcal A_\rho,\mathcal C_\rho)$ is uniformly switching observable
if there are $T_o>0$ and $\beta_o>0$ such that, for every admissible auxiliary
schedule and every $t\ge0$,
\begin{equation}\label{eq:gramian}
 W_o(t,T_o):=\int_t^{t+T_o}\!\Phi_{\mathcal A}(s,t)\T
 \mathcal C_{\rho(s)}\T\mathcal C_{\rho(s)}
 \Phi_{\mathcal A}(s,t)\,ds\succeq\beta_o I.
\end{equation}
\end{definition}

The receiver implements the following observer
\begin{equation}\label{eq:observer}
 \dot{\hat\chi}=\mathcal A_\rho\hat\chi+\mathcal B_\rho\bar x^*
 +L_\rho(y^{z*}-\mathcal C_\rho\hat\chi),
\end{equation}
where
$\hat\chi=\operatorname{col}(\hat z,\hat y^a,\hat x^a)$. Let
$\mathcal A_{e,r}=\mathcal A_r-L_r\mathcal C_r$.

\begin{assumption}\label{as:lifted}
There exist $P_o\succ0$ and $q\in(0,1)$ such that, for every admissible dwell
tuple in every cycle,
\begin{align}\label{eq:monodromy}
 \Psi_{e,k}\T P_o\Psi_{e,k}&\preceq q^2P_o,\nonumber\\
 \Psi_{e,k}&:=
 \exp(\mathcal A_{e,r_m}\Delta_{m,k})\cdots\nonumber\\
 &\hspace{1.7em}\exp(\mathcal A_{e,r_1}\Delta_{1,k}).
\end{align}
The product is time ordered, with the earliest factor on the right.
\end{assumption}

Definition~\ref{def:uco} and Assumption~\ref{as:lifted} have different roles.
Uniform switching observability gives structural identifiability of the
augmented state. Assumption~\ref{as:lifted} is a separate sufficient condition
for the selected mode-dependent gains $L_r$. At cycle boundaries, \eqref{eq:monodromy} is the
standard discrete-time quadratic Lyapunov contraction. Once the gains and $q$ are fixed, the finite dwell
alphabet reduces its verification to a finite family of Linear Matrix Inequalities (LMIs) in $P_o$ and the
factor $q$ can be selected by a scalar search.

The identifiability property in Definition~\ref{def:uco} need not be left to an
a~posteriori Gramian check: it follows from an explicit rank condition on a mode
pair, which also yields a constructive auxiliary design.

Note that the conditions in Definition~\ref{def:uco} and the subsequent
proposition are imposed on the defender-designed auxiliary system rather than
on the physical vehicle dynamics. Since
$A_{z,r}$, $B_{z,r}$, and $C_{z,r}$ are design variables selected by the
defender, the Hurwitz, observability, and DC-gain rank conditions can be
enforced constructively and do not restrict the physical matrices $A_p$ and
$B_p$.

The following proposition provides constructive conditions on the
defender-designed auxiliary system under which the augmented state, including
the injected attack signals, is uniformly identifiable over each complete
switching cycle.

\begin{proposition}\label{prop:pbh}
Suppose every $A_{z,r}$ is Hurwitz, every auxiliary pair
$(A_{z,r},C_{z,r})$ is observable, and define the auxiliary DC-gain matrices
\begin{equation}\label{eq:dcgain}
 G_r:=C_{z,r}A_{z,r}^{-1}B_{z,r}\in\R^{m_z\times3},\qquad r\in\mathcal R.
\end{equation}
If the cycle contains two modes $r,r'$ for which
$\operatorname{rank}(G_r-G_{r'})=3$, and each of those modes has positive dwell
time, then the observability map over every complete cycle is injective.
Moreover, for the finite dwell alphabet, the property in
Definition~\ref{def:uco} holds over a window containing a complete cycle, with
$T_o=2T_\rho$ and a schedule-independent constant $\beta_o>0$.
\end{proposition}

\begin{IEEEproof}
A direction is unobservable over a complete switching cycle only if it produces
zero output on every visited mode interval. For the homogeneous augmented
system used in the observability analysis, the exogenous attack-derivative
input is set to zero. Hence, $\dot{x}^{a}=0$ and $\dot{y}^{a}=0$, so
$x^{a}$ and $y^{a}$ remain constant along the homogeneous trajectory. During
mode $r$, a zero-output trajectory therefore satisfies
\begin{align*}
    \dot{z}=A_{z,r}z-B_{z,r}x^{a},
\qquad
0=C_{z,r}z+y^{a}.
\end{align*}
 Define
$w=z-A_{z,r}^{-1}B_{z,r}x^a$. Then
\[
 \dot w=A_{z,r}w,
 \qquad 0=C_{z,r}w+G_rx^a+y^a.
\]
Using elapsed time $\theta$ from the beginning of the mode-$r$ interval, zero
output gives
\[
 C_{z,r}e^{A_{z,r}\theta}w(0)=-(G_rx^a+y^a)
\]
on a nontrivial interval. Both sides are real analytic in $\theta$, so this
identity extends to every $\theta\ge0$. Since $A_{z,r}$ is Hurwitz, the
left-hand side tends to zero as $\theta\to\infty$; hence
$G_rx^a+y^a=0$. It follows that
$C_{z,r}e^{A_{z,r}\theta}w(0)\equiv0$, and observability of
$(A_{z,r},C_{z,r})$ gives $w(0)=0$.
Repeating the argument in mode $r'$ gives
$G_{r'}x^a+y^a=0$. Hence $(G_r-G_{r'})x^a=0$, and full column rank yields
$x^a=0$, $y^a=0$, and $z=0$. Thus the complete-cycle observability map is
injective.

For an admissible dwell tuple $d$, let $W_c(d)$ denote the observability
Gramian from one cycle boundary to the next. The preceding injectivity gives
$W_c(d)\succ0$. Since $\mathcal D_1\times\cdots\times\mathcal D_m$ is finite,
\[
 \beta_c:=\min_d\lambda_{\min}(W_c(d))>0.
\]
For an arbitrary $t$, let $b\ge t$ be the next cycle boundary. Then
$b-t\le T_\rho$, and the complete cycle starting at $b$ ends no later than
$t+2T_\rho$. With $a:=\max_{r\in\mathcal R}\norm{\mathcal A_r}$, the standard
transition bound gives
$\sigma_{\min}(\Phi_{\mathcal A}(b,t))\ge e^{-aT_\rho}$. The contribution of
that complete cycle to the Gramian therefore satisfies
\[
 W_o(t,2T_\rho)\succeq
 \Phi_{\mathcal A}(b,t)\T W_c(d)\Phi_{\mathcal A}(b,t)
 \succeq e^{-2aT_\rho}\beta_c I.
\]
Thus the uniform switching observability property in
Definition~\ref{def:uco} holds over the initial phase, the dwell tuple, and
changes of dwell tuple between successive cycles, with
$\beta_o=e^{-2aT_\rho}\beta_c$. This completes the proof of the proposition.
\end{IEEEproof}

A simple auxiliary-system construction satisfying the conditions of
Proposition~\ref{prop:pbh} is given below.
\begin{enumerate}
 \item Set $n_z=m_z=3$ and $C_{z,r}=I_3$.
 \item Select $\alpha_r>0$ for every auxiliary mode and choose at least two
 modes with distinct scalars $g_r$.
 \item Set
 \begin{equation}\label{eq:explicit-aux}
  A_{z,r}=-\alpha_rI_3,\qquad
  B_{z,r}=\alpha_rg_rI_3.
 \end{equation}
 \item Choose an auxiliary switching cycle that visits at least two modes with
 distinct $g_r$ and positive dwell times.
\end{enumerate}

To verify that this construction satisfies Proposition~\ref{prop:pbh}, note
first that each $A_{z,r}$ is Hurwitz and each pair
$(A_{z,r},C_{z,r})$ is observable. In this construction,
$z,y^a,x^a\in\R^3$ and $\chi\in\R^9$. Moreover,
\[
 G_r=C_{z,r}A_{z,r}^{-1}B_{z,r}=-g_rI_3,
 \qquad
 G_r-G_{r'}=(g_{r'}-g_r)I_3.
\]
Distinct $g_r$ and $g_{r'}$ therefore give
$\operatorname{rank}(G_r-G_{r'})=3$, so the conditions of
Proposition~\ref{prop:pbh} are satisfied and the property in
Definition~\ref{def:uco} follows.

Because the single-mode augmented pairs $(\mathcal A_r,\mathcal C_r)$ are not
observable, the gains $L_r$ are obtained by lifted periodic-observer synthesis
on the monodromy $\Psi_{e,k}$, with cross-mode observability supplied by
Proposition~\ref{prop:pbh}. For the gains used in Section~\ref{sec:sim}, we
minimize the contraction factor $q$ over $(L_r,P_o,q)$ and then verify
\eqref{eq:monodromy} for all nine admissible dwell tuples generated by
$\mathcal D_\ell=\{0.30,0.35,0.40\}$~s. The resulting monodromy spectral radii
lie in $[0.53,0.79]$, yielding $q=0.79$ and
$\lambda_e=0.29~\mathrm{s}^{-1}$; hence, Assumption~\ref{as:lifted} holds over
the complete admissible dwell family.

The same construction also ensures bounded auxiliary states for bounded
moving-frame inputs. Define
$\alpha:=\min_{r\in\mathcal R}\alpha_r>0$ and
$\bar B_z:=\max_{r\in\mathcal R}\norm{B_{z,r}}$. Then
\begin{equation}\label{eq:auxiliary-bound}
 \norm{z_j(t)}
 \le e^{-\alpha(t-t_0)}\norm{z_j(t_0)}
 +\frac{\bar B_z}{\alpha}
 \sup_{\tau\in[t_0,t]}\norm{\bar x_j(\tau)}.
\end{equation}
Thus the moving-frame input and auxiliary state are bounded under the
conditions following \eqref{eq:moving-frame-bound}.

The next lemma converts the cycle-wise contraction certificate into a
continuous-time exponential bound on the observer transition matrix.
\begin{lemma}\label{lem:observer-decay}
Under Assumption~\ref{as:lifted}, the transition matrix $\Phi_e(t,s)$ of
$\dot\eta=\mathcal A_{e,\rho(t)}\eta$ satisfies
\begin{equation}\label{eq:ues}
 \norm{\Phi_e(t,s)}\le M_e e^{-\lambda_e(t-s)},
 \qquad \lambda_e=-\frac{\ln q}{T_\rho}>0.
\end{equation}
with one admissible choice being
\[
 M_e=K_e^2
 \sqrt{\frac{\lambda_{\max}(P_o)}{\lambda_{\min}(P_o)}}q^{-2},
\]
where $K_e\ge1$ uniformly bounds, in the induced Euclidean norm, every
transition over a subinterval contained within one auxiliary cycle.
\end{lemma}

\begin{IEEEproof}
At cycle boundaries, iteration of \eqref{eq:monodromy} contracts the
$P_o$-norm by $q$ per complete cycle. Any interval $[s,t]$ consists of at most
two partial-cycle pieces and a number $n$ of complete cycles. The partial
pieces contribute at most $K_e^2$. Because every cycle lasts at most $T_\rho$,
$t-s\le(n+2)T_\rho$, and therefore
$n\ge (t-s)/T_\rho-2$. Norm equivalence and
$q^n\le q^{-2}\exp((\ln q)(t-s)/T_\rho)$ give \eqref{eq:ues}. This completes the proof of the lemma.
\end{IEEEproof}

The following theorem establishes an attack-estimation bound, affected by the
attack rates rather than their amplitudes.
\begin{theorem}\label{thm:estimation}
Under Assumptions~\ref{as:attack} and~\ref{as:lifted}, let
$\tilde\chi_{ij}=\chi_{ij}-\hat\chi_{ij}$. For all $t\ge t_0$,
\begin{equation}\label{eq:est-bound}
 \norm{\tilde\chi_{ij}(t)}
 \le M_e e^{-\lambda_e(t-t_0)}\norm{\tilde\chi_{ij}(t_0)}
 +\frac{M_e\norm D}{\lambda_e}\bar d_{ij}.
\end{equation}
Consequently,
\begin{equation}\label{eq:est-limsup}
 \limsup_{t\to\infty}\norm{\tilde x^a_{ij}(t)}
 \le c_e\bar d_{ij},
 \qquad c_e:=\frac{M_e\norm D}{\lambda_e}.
\end{equation}
If both attacks are constant, then $\nu=0$ almost everywhere and the complete
augmented estimation error converges exponentially to zero.
\end{theorem}

\begin{IEEEproof}
Subtracting \eqref{eq:observer} from \eqref{eq:aug-system} gives
$\dot{\tilde\chi}=\mathcal A_{e,\rho}\tilde\chi+D\nu$. Variation of constants
therefore gives
\[
 \tilde\chi(t)=\Phi_e(t,t_0)\tilde\chi(t_0)
 +\int_{t_0}^{t}\Phi_e(t,s)D\nu(s)\,ds.
\]
Taking norms and using Lemma~\ref{lem:observer-decay} and
\eqref{eq:attack-rate} yields
\[
 \begin{aligned}
 \norm{\tilde\chi(t)}
 &\le M_e e^{-\lambda_e(t-t_0)}\norm{\tilde\chi(t_0)}\\
 &\quad+M_e\norm D\bar d_{ij}
 \int_{t_0}^t e^{-\lambda_e(t-s)}\,ds.
 \end{aligned}
\]
The convolution integral satisfies
\[
 \int_{t_0}^t e^{-\lambda_e(t-s)}\,ds
 =\frac{1-e^{-\lambda_e(t-t_0)}}{\lambda_e}
 \le\frac{1}{\lambda_e}.
\]
Substitution proves \eqref{eq:est-bound}; taking the limit superior and
selecting the $x^a$ block gives \eqref{eq:est-limsup}. If $\nu=0$, only the
exponentially decaying term remains, which shows that the estimation error goes to zero asymptotically. This completes the proof of the theorem.
\end{IEEEproof}

\begin{remark}
Note that the observer does not require knowledge of $\bar d_{ij}$ and does not assume a bound on
$x^a_{ij}$ or $y^a_{ij}$. If both attack signals are constant, then
$\bar d_{ij}=0$ can be used and the estimation error converges exponentially
to zero. For slowly varying attacks, the
certified ultimate estimation-error bound decreases proportionally with the
attack-rate bound. Thus boundedness is guaranteed for each finite rate bound
$\bar d_{ij}$, and the ultimate error grows linearly in $\bar d_{ij}$ through
$c_e$.
\end{remark}

\section{Stability Analysis}
\label{sec:stability}

This section analyzes the closed-loop platoon from two complementary
perspectives. Subsection~\ref{sec:resilient} establishes internal stability of
the attack-compensated switched platoon during nominal constant-speed cruising.
Subsection~\ref{sec:string} studies acceleration propagation along a
predecessor-following chain. The two analyses use different leader-motion
conditions. The internal-stability analysis assumes a constant-speed leader,
whereas the string-stability analysis treats the leader acceleration as an
external input.

\subsection{Internal Stability}
\label{sec:resilient}

The internal-stability analysis uses the following nominal cruising
condition.
\begin{assumption}\label{as:nominal-cruising}
For the internal-stability analysis, $u_0\equiv0$, $a_0(0)=0$, and
$v_c=v_0^\star$. Hence $a_0\equiv0$ and $v_0\equiv v_0^\star$.
\end{assumption}

\begin{remark}
While we have considered zero acceleration for the leader in this subsection, nonzero leader acceleration and leader maneuvers are considered
separately in Subsection~\ref{sec:string}, where the leader acceleration is
treated as an external input and its propagation through the platoon is
analyzed.
\end{remark}

Let $d_{i-1,i}>0$ and $h>0$ denote the standstill gap and time headway. Define
$d^h_{i-1,i}=d_{i-1,i}+hv_0^\star$ and
$d^h_{0,i}=\sum_{\ell=1}^i d^h_{\ell-1,\ell}$. The nominal reference and
deviation coordinates are
\begin{align}\label{eq:reference}
x_{i,\mathrm{ref}}&=
\begin{bmatrix}s_0-d^h_{0,i}&v_0^\star&0\end{bmatrix}\T,\nonumber\\
e_i&=x_i-x_{i,\mathrm{ref}},\qquad e_0=0.
\end{align}
Under Assumption~\ref{as:nominal-cruising},
$\dot e_i=A_\sigma e_i+B_\sigma u_i$. With
$k_p=[k_{1,p},k_{2,p},k_{3,p}]\T$ and $e_v=[0,1,0]\T$, consider
\begin{equation}\label{eq:nominal}
u_i=-k_\sigma\T\sum_{j\in\mathcal I_i}(e_i-e_j)
-k_{1,\sigma}h e_v\T e_i.
\end{equation}
Define the headway-augmented matrix
\begin{equation}\label{eq:Abar}
\bar A_p=A_p-k_{1,p}hB_pe_v\T.
\end{equation}

For the predecessor chain, $\mathcal I_i=\{i-1\}$, the reference quantities
cancel and \eqref{eq:nominal} becomes
\begin{align}\label{eq:local-controller}
u_i={}&k_{1,\sigma}(s_{i-1}-s_i-d_{i-1,i}-hv_i)\nonumber\\
&+k_{2,\sigma}(v_{i-1}-v_i)\nonumber\\
&+k_{3,\sigma}(a_{i-1}-a_i).
\end{align}
Thus this implementation requires only predecessor neighbor states.

With per-link attack estimates available, \eqref{eq:nominal} is made resilient
by compensating each received state. The controller mitigates the effects of the attack on the received
state as
$x_{ij}^*-\hat x^a_{ij}=x_j+\tilde x^a_{ij}$; hence, the compensated
platoon is driven by the reconstruction errors.
Replacing each received state in \eqref{eq:nominal} by its compensated value
gives, in deviation coordinates,
\begin{equation}\label{eq:res-controller}
u_i=-\sum_{j\in\mathcal I_i}k_\sigma\T
(e_i-e_j-\tilde x^a_{ij})-k_{1,\sigma}he_v\T e_i.
\end{equation}
Set $r_i=\sum_{j\in\mathcal I_i}\tilde x^a_{ij}$,
$R=[r_1\T,\ldots,r_N\T]\T$, and
$X=[e_1\T,\ldots,e_N\T]\T$. Then
\begin{equation}\label{eq:stacked-loop}
\dot X=A_{c,\sigma}X+G_\sigma R,
\end{equation}
where
\begin{equation}\label{eq:closed-matrices}
A_{c,p}=I_N\otimes\bar A_p-(L+P)\otimes B_pk_p\T,
\qquad G_p=I_N\otimes B_pk_p\T.
\end{equation}

The next lemma gives necessary and sufficient gain conditions for stability
of each frozen physical mode.
\begin{lemma}\label{lem:frozen-plant}
For fixed $p$, $A_{c,p}$ is Hurwitz if and only if, for every eigenvalue
$\lambda_i$ of $L+P$,
\begin{align}\label{eq:routh-graph}
&k_{1,p}>0,\quad k_{1,p}h+\lambda_i k_{2,p}>0,\quad
1+\lambda_i k_{3,p}>0,\nonumber\\
&(1+\lambda_i k_{3,p})(k_{1,p}h+\lambda_i k_{2,p})
>\tau_p\lambda_i k_{1,p}.
\end{align}
\end{lemma}

\begin{IEEEproof}
Let $H=L+P$. Since its eigenvalues are real, there is a nonsingular
$T\in\R^{N\times N}$ such that
$J=T^{-1}HT$ is upper triangular with diagonal entries
$\lambda_1,\ldots,\lambda_N$. Under the similarity transformation
$\zeta=(T^{-1}\otimes I_3)X$,
\begin{align*}
&(T^{-1}\otimes I_3)A_{c,p}(T\otimes I_3)\\
&\qquad=I_N\otimes\bar A_p-J\otimes B_pk_p\T.
\end{align*}
The transformed matrix is block upper triangular, with diagonal blocks
\[
F_{i,p}:=\bar A_p-\lambda_iB_pk_p\T
=\begin{bmatrix}
0&1&0\\
0&0&1\\
-c_{0,i,p}&-c_{1,i,p}&-c_{2,i,p}
\end{bmatrix}.
\]
Here
\[
\begin{aligned}
c_{0,i,p}&=\lambda_i k_{1,p}/\tau_p,\\
c_{1,i,p}&=(k_{1,p}h+\lambda_i k_{2,p})/\tau_p,\\
c_{2,i,p}&=(1+\lambda_i k_{3,p})/\tau_p.
\end{aligned}
\]
Similarity preserves eigenvalues, and the spectrum of a block
upper-triangular matrix is the union of the spectra of its diagonal blocks.
Consequently, $A_{c,p}$ is Hurwitz if and only if every displayed block is
Hurwitz. The characteristic polynomial of block $i$ is
\[
s^3+c_{2,i,p}s^2+c_{1,i,p}s+c_{0,i,p}.
\]
For a monic cubic $s^3+a_2s^2+a_1s+a_0$, the Routh--Hurwitz conditions are
$a_2>0$, $a_1>0$, $a_0>0$, and $a_2a_1>a_0$. Since
$\tau_p>0$ and $\lambda_i>0$, substituting the three coefficients gives
exactly \eqref{eq:routh-graph}. This completes the proof of the lemma.
\end{IEEEproof}

The resulting switched-platoon bound is stated next in terms of the certified
rate-dependent estimation errors.
\begin{theorem}\label{thm:platoon}
Suppose there are $P_p^c\succ0$, $\lambda_c>0$, and $\mu_c\ge1$ such that
\begin{equation}\label{eq:plant-cert}
A_{c,p}\T P_p^c+P_p^cA_{c,p}\preceq-\lambda_cP_p^c,
\qquad P_p^c\preceq\mu_cP_q^c
\end{equation}
for all $p,q\in\mathcal S$. Assume in addition that the physical ADT satisfies
\begin{equation}\label{eq:plant-adt}
\tau_a>\frac{2\ln\mu_c}{\lambda_c},
\qquad \xi_c:=\frac{\lambda_c}{2}-\frac{\ln\mu_c}{\tau_a}>0.
\end{equation}
Here $\tau_a$ is the average dwell time introduced in \eqref{eq:adt}.
Then the compensated platoon is Uniformly, Ultimately Bounded (UUB). More precisely, with $N_0$ the chatter
bound of \eqref{eq:adt}, let
\begin{align}\label{eq:constants-control}
m_c&:=\min_p\lambda_{\min}(P_p^c),\qquad
\kappa_c:=\max_p\lambda_{\max}(G_p\T P_p^cG_p),\nonumber\\
\bar R_a&:=c_e^2\sum_{i=1}^N d_i
\sum_{j\in\mathcal I_i}\bar d_{ij}^{\,2},\qquad d_i:=|\mathcal I_i|.
\end{align}
Then
\begin{equation}\label{eq:platoon-radius}
\limsup_{t\to\infty}\norm{X(t)}^2
\le\frac{2\mu_c^{N_0}\kappa_c}
{m_c\lambda_c\xi_c}\,\bar R_a.
\end{equation}
\end{theorem}

\begin{IEEEproof}
Define the mode-dependent Lyapunov function
$V_{\sigma(t)}(X)=X\T P_{\sigma(t)}^cX$. On an interval where
$\sigma(t)=p$, write $V_p=X\T P_p^cX$. Since $P_p^c$ is constant within that
interval, \eqref{eq:stacked-loop} gives
\begin{align*}
\dot V_p
&=X\T\!\left(A_{c,p}\T P_p^c+P_p^cA_{c,p}\right)X
+2X\T P_p^cG_pR\\
&\le-\lambda_cV_p+2X\T P_p^cG_pR,
\end{align*}
where the inequality follows from \eqref{eq:plant-cert}. Apply Young's
inequality to
$a=(P_p^c)^{1/2}X$ and
$b=(P_p^c)^{1/2}G_pR$ with parameter $\lambda_c/2$:
\[
2a\T b\le\frac{\lambda_c}{2}\norm a^2
+\frac{2}{\lambda_c}\norm b^2.
\]
Using the definition of $\kappa_c$ in \eqref{eq:constants-control} yields
\begin{equation}\label{eq:plant-flow-bound}
\dot V_p\le-\frac{\lambda_c}{2}V_p
+\frac{2\kappa_c}{\lambda_c}\norm R^2.
\end{equation}

At a switch from mode $p$ to mode $q$, the state $X$ is continuous, while
\eqref{eq:plant-cert} gives
\[
V_q^+=X\T P_q^cX\le\mu_cX\T P_p^cX=\mu_cV_p^-.
\]
Iterating \eqref{eq:plant-flow-bound} over the flow intervals introduces one
factor $\mu_c$ at each switch. By \eqref{eq:adt},
\[
\mu_c^{N_\sigma(t,s)}
e^{-\lambda_c(t-s)/2}
\le\mu_c^{N_0}e^{-\xi_c(t-s)}.
\]
Consequently,
\begin{align}\label{eq:adt-comparison}
V_{\sigma(t)}(t)
&\le\mu_c^{N_0}e^{-\xi_c(t-t_0)}
V_{\sigma(t_0)}(t_0)\nonumber\\
&\quad+\frac{2\mu_c^{N_0}\kappa_c}{\lambda_c}
\int_{t_0}^t e^{-\xi_c(t-s)}\norm{R(s)}^2\,ds.
\end{align}

Furthermore, for every follower, Cauchy--Schwarz gives
\[
\left|\sum_{j\in\mathcal I_i}\tilde x^a_{ij}\right|^2
\le d_i\sum_{j\in\mathcal I_i}\norm{\tilde x^a_{ij}}^2.
\]
Summing over $i$, applying Theorem~\ref{thm:estimation}, and using the
finiteness of $\mathcal E_c$ gives
\[
\limsup_{t\to\infty}\norm{R(t)}^2\le\bar R_a.
\]
Fix $\varepsilon>0$. There is therefore a time $T_\varepsilon$ such that
$\norm{R(t)}^2\le\bar R_a+\varepsilon$ for all
$t\ge T_\varepsilon$. Set
\[
c_\varepsilon:=
\frac{2\kappa_c}{\lambda_c}(\bar R_a+\varepsilon).
\]
Applying \eqref{eq:adt-comparison} from $T_\varepsilon$ to $t$ and evaluating
the convolution integral gives
\begin{align*}
V_{\sigma(t)}(t)
&\le\mu_c^{N_0}e^{-\xi_c(t-T_\varepsilon)}
V_{\sigma(T_\varepsilon)}(T_\varepsilon)
+\mu_c^{N_0}c_\varepsilon
\int_{T_\varepsilon}^t e^{-\xi_c(t-s)}\,ds\\
&=\mu_c^{N_0}
\left(V_{\sigma(T_\varepsilon)}(T_\varepsilon)
-\frac{c_\varepsilon}{\xi_c}\right)
e^{-\xi_c(t-T_\varepsilon)}
+\frac{\mu_c^{N_0}c_\varepsilon}{\xi_c}.
\end{align*}
For each fixed $\varepsilon>0$, letting $t\to\infty$ eliminates the transient
term and gives
$
\limsup_{t\to\infty}V_{\sigma(t)}(t)
\le
\frac{2\mu_c^{N_0}\kappa_c}
{\lambda_c\xi_c}
\left(\bar R_a+\varepsilon\right).
$
Since this inequality holds for every $\varepsilon>0$, letting
$\varepsilon\to0^+$ yields
$
\limsup_{t\to\infty}V_{\sigma(t)}(t)
\le
\frac{2\mu_c^{N_0}\kappa_c}
{\lambda_c\xi_c}\bar R_a.
$

Finally,
$V_{\sigma(t)}(X)\ge m_c\norm X^2$, which proves
\eqref{eq:platoon-radius}. The observer and switched-system estimates used
above are uniform; hence, on every bounded set of initial plant and observer
errors, $T_\varepsilon$ can be chosen uniformly. This proves the stated
entry-time form of UUB. If all attacks are constant, Theorem~\ref{thm:estimation} gives
$\norm{R(t)}^2\le c\,e^{-2\lambda_e(t-t_0)}$ for some $c>0$. Substituting into
\eqref{eq:adt-comparison}, the convolution
$\int_{t_0}^t e^{-\xi_c(t-s)}\norm{R(s)}^2\,ds$ decays at rate
$\min(\xi_c,2\lambda_e)$; hence $V_{\sigma(t)}(t)\to0$ and, by
$V_{\sigma(t)}(X)\ge m_c\norm X^2$, $X(t)\to0$ exponentially. When
$\xi_c=2\lambda_e$, the convolution produces a $t\,e^{-\xi_c t}$ term, which is
bounded by $e^{-\lambda' t}$ for any $\lambda'<\xi_c$. This completes the proof of the theorem.
\end{IEEEproof}

\subsection{String Stability}
\label{sec:string}
Unlike the preceding internal-stability analysis, which allows the general
leader--follower communication topology, the string-stability analysis is
restricted to the predecessor-following chain.
 The leader
acceleration $a_0$ is now an external input, so the nominal-cruising condition
of Assumption~\ref{as:nominal-cruising} is not imposed.

\begin{remark}
The estimation-error bound of Theorem~\ref{thm:estimation} does not depend on
boundedness of the moving-frame state $\bar x_j$: subtracting
\eqref{eq:observer} from \eqref{eq:aug-system} cancels the common input
$\bar x^*$, leaving $\dot{\tilde\chi}=\mathcal A_{e,\rho}\tilde\chi+D\nu$. Hence
the residual $r_i=\tilde x^a_{i,i-1}$, and therefore the propagation bound
\eqref{eq:power-string}, holds for any bounded-rate leader input, including
maneuvers with $a_0\not\equiv0$. Boundedness of the transmitted auxiliary
signal $y_j^{z}$ is a separate requirement: by \eqref{eq:auxiliary-bound} it
needs $\bar x_j\in\mathcal L_\infty$, whose position component requires
$v_j-v_c\in\mathcal L_1\cap\mathcal L_\infty$ as in
\eqref{eq:moving-frame-bound}. We therefore restrict the admissible leader
maneuvers to those for which $v_j-v_c\in\mathcal L_1\cap\mathcal L_\infty$ for
every follower.
\end{remark}

Define the physical CTH
spacing error and relative velocity
\begin{align}\label{eq:local-string-state}
\zeta_i&=s_{i-1}-s_i-d_{i-1,i}-hv_i,\nonumber\\
\Delta v_i&=v_{i-1}-v_i,\qquad
\xi_i=\begin{bmatrix}\zeta_i&\Delta v_i&a_i\end{bmatrix}\T.
\end{align}
Writing $r_i=\tilde x^a_{i,i-1}\in\R^3$, the local dynamics are
\begin{align}\label{eq:local-string-dynamics}
\dot\xi_i&=A_p^s\xi_i+B_p^aa_{i-1}+B_p^rr_i,\nonumber\\
a_i&=C_a\xi_i.
\end{align}
Here $p=\sigma(t)$, and
\begin{equation}\label{eq:string-matrices}
A_p^s=\begin{bmatrix}
0&1&-h\\
0&0&-1\\
k_{1,p}/\tau_p&k_{2,p}/\tau_p&-(1+k_{3,p})/\tau_p
\end{bmatrix},
\end{equation}
\begin{align*}
B_p^a&=\begin{bmatrix}0&1&k_{3,p}/\tau_p\end{bmatrix}\T,
\qquad C_a=\begin{bmatrix}0&0&1\end{bmatrix},\\
B_p^r&=\begin{bmatrix}
0&0&0\\
0&0&0\\
k_{1,p}/\tau_p&k_{2,p}/\tau_p&k_{3,p}/\tau_p
\end{bmatrix}.
\end{align*}

Following the input--output definition of $\mathcal L_2$ string stability in
\cite{Ploeg14}, we evaluate acceleration propagation over finite time horizons
and for every admissible switching signal. We require non-amplification rather
than strict attenuation because each internally stable frozen CTH link has
unit acceleration gain at zero frequency.

We next formalize the nominal switched acceleration non-amplification property
used in the string analysis.
\begin{definition}
\label{def:nominal-string}
The predecessor chain is nominally switched $\mathcal L_2$ acceleration string
stable if, for zero local initial states and $r_i\equiv0$,
\begin{equation}\label{eq:nominal-string-def}
\norm{a_i}_{2,T}\le\norm{a_{i-1}}_{2,T}
\end{equation}
for every $i$, every $T>0$, and every admissible physical switching signal.
\end{definition}

The next result establishes nominal string stability and quantifies the
additional acceleration power induced by nonvanishing estimation errors.
\begin{theorem}
\label{thm:string}
Assume there is $P^0\succ0$ and $\alpha_0>0$ such that
\begin{equation}\label{eq:common-stability}
(A_p^s)\T P^0+P^0A_p^s\preceq-\alpha_0P^0,
\qquad p\in\mathcal S.
\end{equation}
Suppose also that there are common matrices $P^a\succ0$, $P^r\succ0$ and a
constant $\gamma_r>0$ satisfying the following inequalities for every $p$.
Here $\He(M):=M+M\T$.
\begin{equation}\label{eq:lmi-nominal}
\begin{bmatrix}
\He(P^aA_p^s)+C_a\T C_a&P^aB_p^a\\
(B_p^a)\T P^a&-1
\end{bmatrix}\preceq0,
\end{equation}
\begin{equation}\label{eq:lmi-residual}
\begin{bmatrix}
\He(P^rA_p^s)+C_a\T C_a&P^rB_p^r\\
(B_p^r)\T P^r&-\gamma_r^2I_3
\end{bmatrix}\preceq0.
\end{equation}
Then the nominal platoon satisfies \eqref{eq:nominal-string-def}. With
nonvanishing estimation errors and arbitrary finite initial conditions,
\begin{equation}\label{eq:power-step}
\norm{a_i}_{\mathcal P}
\le\norm{a_{i-1}}_{\mathcal P}+\gamma_r\norm{r_i}_{\mathcal P},
\end{equation}
and consequently
\begin{equation}\label{eq:power-string}
\norm{a_i}_{\mathcal P}
\le\norm{a_0}_{\mathcal P}
+\gamma_r\sum_{\ell=1}^i\norm{r_\ell}_{\mathcal P}.
\end{equation}
By Theorem~\ref{thm:estimation},
$\norm{r_i}_{\mathcal P}\le c_e\bar d_{i,i-1}$.
\end{theorem}

\begin{IEEEproof}
By linearity, decompose the state and acceleration as
\[
\xi_i=\xi_i^a+\xi_i^r+\xi_i^h,\qquad
a_i=a_i^a+a_i^r+a_i^h.
\]
The superscripts $a$, $r$, and $h$ denote, respectively, the zero-state
response driven by $a_{i-1}$, the zero-state response driven by $r_i$, and the
homogeneous response from the local initial state.

First consider the predecessor-acceleration channel and define
$V_i^a=(\xi_i^a)\T P^a\xi_i^a$. On an interval with $\sigma(t)=p$,
\[
\dot V_i^a
=(\xi_i^a)\T\He(P^aA_p^s)\xi_i^a
+2(\xi_i^a)\T P^aB_p^aa_{i-1}.
\]
Premultiplying and postmultiplying \eqref{eq:lmi-nominal} by
$\operatorname{col}(\xi_i^a,a_{i-1})$ gives
\[
\dot V_i^a+|a_i^a|^2-|a_{i-1}|^2\le0.
\]
The matrix $P^a$ is common to all modes, so $V_i^a$ does not jump when
$\sigma$ switches. Integrating over $[0,T]$, using the zero initial state, and
dropping the nonnegative terminal value gives
\[
\int_0^T|a_i^a(t)|^2dt
\le\int_0^T|a_{i-1}(t)|^2dt,
\]
or $\norm{a_i^a}_{2,T}\le\norm{a_{i-1}}_{2,T}$. With
$r_i\equiv0$ and zero local initial state, this proves
\eqref{eq:nominal-string-def}.

For the estimation-error channel, let
$V_i^r=(\xi_i^r)\T P^r\xi_i^r$. Applying
\eqref{eq:lmi-residual} to $\operatorname{col}(\xi_i^r,r_i)$ gives
\[
\dot V_i^r+|a_i^r|^2-\gamma_r^2\norm{r_i}^2\le0.
\]
Again the common storage has no switching jumps. Integration from the zero
initial state yields
$
\norm{a_i^r}_{2,T}\le\gamma_r\norm{r_i}_{2,T}.
$

Finally, \eqref{eq:common-stability} gives
$\dot V_i^h\le-\alpha_0V_i^h$ for
$V_i^h=(\xi_i^h)\T P^0\xi_i^h$ in every mode. Since $P^0$ is common, the
homogeneous response decays exponentially under arbitrary switching.
Consequently, $a_i^h\in\mathcal L_2$ and
$\norm{a_i^h}_{\mathcal P}=0$.

For the complete response, the finite-horizon triangle inequality gives
\[
\norm{a_i}_{2,T}
\le\norm{a_i^a}_{2,T}
+\norm{a_i^r}_{2,T}
+\norm{a_i^h}_{2,T}.
\]
Dividing the finite-horizon inequality by $\sqrt{T}$, taking the limit superior
as $T\to\infty$, and using~\eqref{eq:powernorm} gives
\eqref{eq:power-step}. Iterating this one-step inequality from vehicle $1$ to
vehicle $i$ yields~\eqref{eq:power-string}.
 By
\eqref{eq:est-limsup}, for every $\varepsilon>0$, there exists a finite time
$T_\varepsilon$ such that
$
\norm{r_i(t)}
\le
c_e\bar d_{i,i-1}+\varepsilon,
\qquad
t\ge T_\varepsilon.
$
Therefore, for every $T>T_\varepsilon$,
$
\frac{1}{T}\int_0^T\norm{r_i(t)}^2\,dt
\le
\frac{1}{T}\int_0^{T_\varepsilon}\norm{r_i(t)}^2\,dt
+
\frac{T-T_\varepsilon}{T}
\left(c_e\bar d_{i,i-1}+\varepsilon\right)^2.
$
Taking the limit superior as $T\to\infty$ eliminates the contribution of the
finite initial interval and gives
$
\norm{r_i}_{\mathcal P}
\le
c_e\bar d_{i,i-1}+\varepsilon.
$
Since this inequality holds for every $\varepsilon>0$, letting
$\varepsilon\to0^+$ yields
$
\norm{r_i}_{\mathcal P}
\le
c_e\bar d_{i,i-1}.
$

\end{IEEEproof}

\begin{remark}
Note that a joint inequality for the
input $[a_{i-1},r_i\T]\T$ could also be used. The separate formulation exploits
superposition, preserves the unit gain condition for the predecessor-acceleration channel, and can reduce conservatism in the estimation-error
bound.
\end{remark}

\begin{corollary}
\label{cor:one-attacked-link}
Suppose only link $(k,k-1)$ has a nonzero estimation error, so that
$r_k\ne0$ and $r_i=0$ for $i\ne k$. Then, for every $i\ge k$,
\begin{equation}\label{eq:single-link-propagation}
\norm{a_i}_{\mathcal P}
\le\norm{a_{k-1}}_{\mathcal P}
+\gamma_r\norm{r_k}_{\mathcal P}
\le\norm{a_{k-1}}_{\mathcal P}
+\gamma_rc_e\bar d_{k,k-1}.
\end{equation}
Moreover,
$\norm{a_i}_{\mathcal P}\le\norm{a_{i-1}}_{\mathcal P}$ for every
$i>k$.
\end{corollary}

\begin{IEEEproof}
At vehicle $k$, since $r_k\ne0$, \eqref{eq:power-step} gives
\[
\norm{a_k}_{\mathcal P}\le\norm{a_{k-1}}_{\mathcal P}
+\gamma_r\norm{r_k}_{\mathcal P}.
\]
For each $i>k$, $r_i=0$, so \eqref{eq:power-step} reduces to
$\norm{a_i}_{\mathcal P}\le\norm{a_{i-1}}_{\mathcal P}$. Chaining these
inequalities from $k$ to $i$ yields
$\norm{a_i}_{\mathcal P}\le\norm{a_{k-1}}_{\mathcal P}
+\gamma_r\norm{r_k}_{\mathcal P}$ for every $i\ge k$. Finally,
Theorem~\ref{thm:estimation} bounds
$\norm{r_k}_{\mathcal P}\le c_e\bar d_{k,k-1}$, which proves
\eqref{eq:single-link-propagation}.
\end{IEEEproof}

Thus a nonzero estimation error can increase acceleration across the affected
link, so the nominal result in Definition~\ref{def:nominal-string} does not
apply across that link.
After that link, however, the effect is not further amplified in the certified
power norm. This statement does not bound instantaneous acceleration peaks,
which would require an $\mathcal L_\infty$ analysis. If the link attacks are
constant, $\bar d_{k,k-1}=0$ and the estimation error converges exponentially
to zero after the observer transient. A smaller nonzero attack-rate bound gives a proportionally smaller certified contribution in
\eqref{eq:single-link-propagation}.

\begin{remark}
The nominal statement \eqref{eq:nominal-string-def} is string-length
independent, that is, it holds with constants independent of the number of
vehicles $N$. The bound \eqref{eq:power-string} extends this to
the perturbed setting with an explicit, computable dependence on the
estimation errors of the upstream links. This yields a length-independence criterion for the derived certificate: the
bound \eqref{eq:power-string} is uniform in $N$ if and only if the per-link
error powers are spatially summable,
$\sup_{i\ge1}\sum_{\ell=1}^i\norm{r_\ell}_{\mathcal P}<\infty,$
which holds, for instance, whenever
$\sum_{\ell\ge1}\bar d_{\ell,\ell-1}<\infty$. The certified attack-rate bounds
$\bar d_{\ell,\ell-1}$ thus directly determine whether the certified bound is
string-length independent. Spatial summability is established as necessary and
sufficient for uniformity of this upper bound; it is not claimed necessary for
the true closed-loop acceleration response.
\end{remark}

For a fixed mode, the nominal transfer function in
\eqref{eq:local-string-dynamics} is
\begin{align}\label{eq:frozen-transfer}
H_p(s)&=\frac{a_i(s)}{a_{i-1}(s)}\nonumber\\
&=
\frac{k_{3,p}s^2+k_{2,p}s+k_{1,p}}
{\tau_ps^3+(1+k_{3,p})s^2
+(k_{2,p}+k_{1,p}h)s+k_{1,p}}.
\end{align}

The following lemma gives an exact frozen-mode frequency-domain test for
acceleration non-amplification, i.e., for $\norm{H_p}_\infty\le1$.
\begin{lemma}\label{lem:frequency}
Assume the denominator of \eqref{eq:frozen-transfer} is Hurwitz, equivalently
\begin{align}\label{eq:local-routh}
&k_{1,p}>0,\quad 1+k_{3,p}>0,\quad k_{2,p}+k_{1,p}h>0,\nonumber\\
&\hspace{1em}
(1+k_{3,p})(k_{2,p}+k_{1,p}h)>\tau_pk_{1,p}.
\end{align}
Define
\begin{align}\label{eq:q-coefficients}
q_{0,p}&=k_{1,p}(k_{1,p}h^2+2hk_{2,p}-2),\nonumber\\
q_{1,p}&=1+2k_{3,p}-2\tau_p(k_{2,p}+k_{1,p}h).
\end{align}
Then $\norm{H_p}_\infty\le1$ if and only if
\begin{equation}\label{eq:frequency-test}
\begin{cases}
q_{0,p}\ge0, & q_{1,p}\ge0,\\
4\tau_p^2q_{0,p}\ge q_{1,p}^2, & q_{1,p}<0.
\end{cases}
\end{equation}
\end{lemma}

\begin{IEEEproof}
Since the denominator of \eqref{eq:frozen-transfer} is Hurwitz,
$\norm{H_p}_\infty\le1$ holds if and only if
$|N_p(j\omega)|^2\le|D_p(j\omega)|^2$ for all $\omega$, where $N_p$ and $D_p$
are the numerator and denominator of \eqref{eq:frozen-transfer}. Taking squared
magnitudes and collecting terms gives
\[
|D_p(j\omega)|^2-|N_p(j\omega)|^2
=\omega^2\bigl(\tau_p^2y^2+q_{1,p}y+q_{0,p}\bigr),\qquad y=\omega^2,
\]
with $q_{0,p},q_{1,p}$ as in \eqref{eq:q-coefficients}. The factor $\omega^2$ is
nonnegative, so the condition reduces to
$g(y):=\tau_p^2y^2+q_{1,p}y+q_{0,p}\ge0$ for all $y\ge0$.

Because $\tau_p^2>0$, $g$ is an upward parabola with vertex at
$y^\star=-q_{1,p}/(2\tau_p^2)$. If $q_{1,p}\ge0$, then $y^\star\le0$, so $g$ is
smallest at $y=0$ and $g\ge0$ on $[0,\infty)$ exactly when $g(0)=q_{0,p}\ge0$.
If $q_{1,p}<0$, then $y^\star>0$, so the smallest value is
$g(y^\star)=q_{0,p}-q_{1,p}^2/(4\tau_p^2)$, which is nonnegative exactly when
$4\tau_p^2q_{0,p}\ge q_{1,p}^2$. These are precisely the two cases in
\eqref{eq:frequency-test}. This completes the proof of the lemma.
\end{IEEEproof}

\begin{remark}
At $h=0$, $q_{0,p}=-2k_{1,p}<0$, so every stabilizing constant-offset
predecessor law amplifies some sufficiently low frequency. A positive headway
can make $q_{0,p}\ge0$ and thereby enable non-amplification. This explains the
mechanism by which the CTH term prevents low-frequency growth. The frozen test
\eqref{eq:frequency-test} is necessary and sufficient mode by mode; the common
LMIs \eqref{eq:common-stability}--\eqref{eq:lmi-residual} are stronger because
they certify the switched system without an additional dwell-time restriction.
\end{remark}

\section{Simulation Studies}\label{sec:sim}
\subsection{Numerical setup}
The simulations consider one leader and $N=7$ followers. The powertrain modes
have $\tau_p\in\{0.3,0.5,0.7\}$~s, the standstill distance is $5$~m, the
cruising speed is $20$~m/s, and the CTH is $h=1.2$~s; hence, the equilibrium
gap is $29$~m. The controller gains are mode dependent:
$k_1=[0.47,0.99,0.44]\T$, $k_2=[0.50,1.00,0.50]\T$, and
$k_3=[0.53,1.01,0.70]\T$, corresponding respectively to
$\tau_1=0.3$~s, $\tau_2=0.5$~s, and $\tau_3=0.7$~s. The vehicles synchronously
cycle through the three powertrain modes with a $2$~s dwell in each mode. Thus,
mode $p\in\{1,2,3\}$ is active on
$[6k+2(p-1),6k+2p)$, $k\in\mathbb N_0$. This schedule is selected only to
exercise all three modes; the common physical certificate used in
Theorems~\ref{thm:platoon} and~\ref{thm:string} permits arbitrary physical
switching.

The auxiliary system uses $n_z=m_z=3$ and two modes with diagonal
$A_{z,r}$, $B_{z,r}$, and $C_{z,r}$, chosen so that each $A_{z,r}$ is Hurwitz,
each pair $(A_{z,r},C_{z,r})$ is observable, and the two dc-gain matrices differ
with $\operatorname{rank}(G_1-G_2)=3$, as required by
Proposition~\ref{prop:pbh}. The mode-dependent observer gains
$L_r=\operatorname{col}(L_r^z,L_r^y,L_r^x)$ are likewise diagonal. All numerical
values are provided in the public repository.
All auxiliary systems and observers use the right-continuous cycle $\rho=1$
then $\rho=2$, with each mode dwell drawn from the admissible set
$\mathcal D_\ell=\{0.30,0.35,0.40\}$~s; the reported run realizes the $0.40$~s
dwell, so $\rho=1$ on $[0.8k,0.8k+0.4)$ and $\rho=2$ on $[0.8k+0.4,0.8k+0.8)$.
The maximum cycle period is $T_\rho=0.8$~s, and the observability horizon is
$2T_\rho=1.6$~s.

For the physical platoon, a common certificate $P_p^c=P^c$ is feasible with
$\lambda_c=0.25~\mathrm{s}^{-1}$ and $\mu_c=1$. Hence the threshold in
\eqref{eq:plant-adt} is zero and $\xi_c=0.125~\mathrm{s}^{-1}$, so no positive
ADT lower bound is imposed. The worst maximum-eigenvalue residual in
\eqref{eq:plant-cert} is $-2.67\times10^{-2}$, providing strict slack and
verifying the common physical certificate for arbitrary switching among the
three certified plant--controller mode pairs $(\tau_p,k_p)$. The complete
certificate and simulation parameters,
including the matrices, noise model, and initial conditions, are available in
the public repository.\footnote{\url{https://github.com/AlienEslami/Resilient-Switched-cth-Platoons-code}}

\subsection{All-link attacks, estimation, and resilient response}
All attacks start at $t_a=8$~s. To avoid an artificial inconsistency among
the corrupted position, velocity, and acceleration fields, the state-packet
attack is generated from a scalar false position trajectory $q(t)$ as
\begin{equation}\label{eq:sim-consistent-packet}
 x^a_{i,i-1}(t)=\ell_i[q(t),\dot q(t),\ddot q(t)]\T.
\end{equation}
The onset is smoothed by $s_\Delta(t)$, a fifth-order transition from zero to
one over the interval $[t_a,t_a+\Delta]$. Thus,
\eqref{eq:sim-consistent-packet} remains
kinematically consistent during activation as well as afterward. The
auxiliary-output injection is a separate communication-channel signal and is
not subject to vehicle kinematics. Two attacks are applied simultaneously
to all seven predecessor links.

\subsubsection{Low-frequency attack}
For the low-frequency case,
\begin{align}\label{eq:sim-low-frequency}
 q_\ell(t)&=20s_{12}(t)\sin\!\big(0.05(t-t_a)\big),\nonumber\\
 y^a_{i,i-1}&=\ell_i s_{12}(t)[-4,1.4,0.8]\T
 \sin\!\big(0.05(t-t_a)\big),
\end{align}
where $\ell_i=1+0.05(i-1)$, and
\eqref{eq:sim-consistent-packet} is evaluated analytically with $q=q_\ell$.
After the activation transient, the position amplitude is large, while the
steady sinusoidal velocity and acceleration amplitudes are only $1$~m/s and
$0.05$~m/s$^2$, respectively, before link scaling. The larger short-lived
derivative terms near $t_a$ in
Fig.~\ref{fig:low-frequency} are produced by the smooth activation itself.
The figure shows the result on link $(4,3)$.

\begin{figure}[!t]
\centering
\includegraphics[width=\columnwidth]{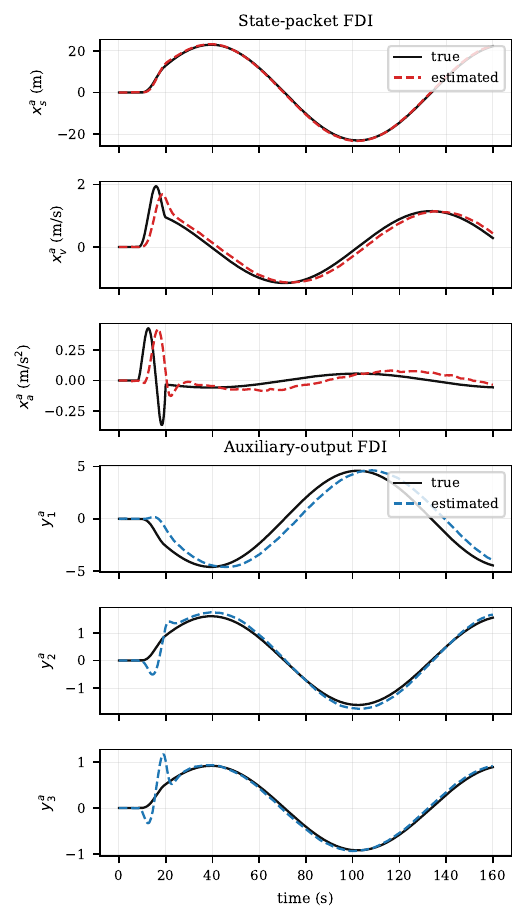}
\caption{Low-frequency all-link attack: injected and estimated signals on
representative link $(4,3)$.}
\label{fig:low-frequency}
\end{figure}

The uncompensated attack generates large oscillatory spacing and velocity
errors in Fig.~\ref{fig:response-low-frequency}. The proposed compensation
reduces these errors sharply and maintains positive distances despite the
persistent reconstruction lag. This is the expected rate-dependent behavior:
a slowly varying attack produces a smaller residual than a faster attack of
the same amplitude.

\begin{figure}[!t]
\centering
\includegraphics[width=\columnwidth]{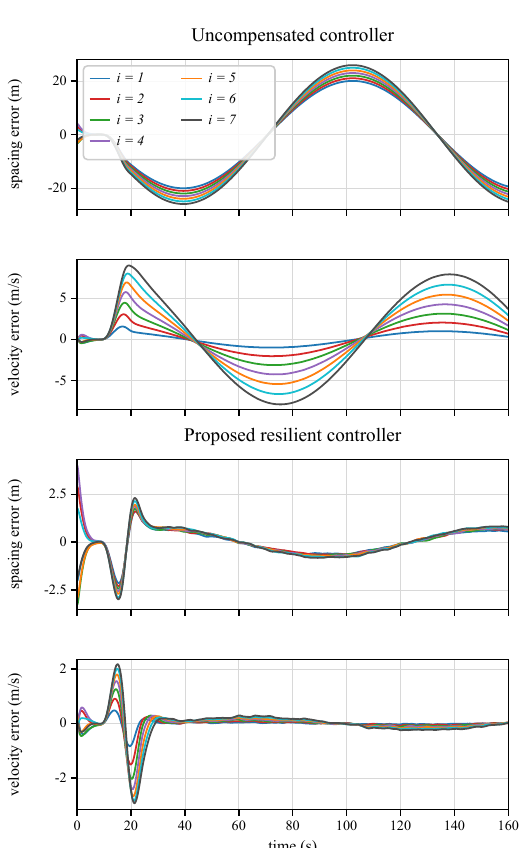}
\caption{Low-frequency all-link attack: uncompensated and resilient platoon
responses.}
\label{fig:response-low-frequency}
\end{figure}

\subsubsection{Ramp attack}
For the unbounded ramp attack,
\begin{equation}\label{eq:sim-ramp}
 q_r(t)=0.35s_2(t)(t-t_a),\qquad y^a_{i,i-1}=0,
\end{equation}
and \eqref{eq:sim-consistent-packet} is used with $q=q_r$.
This state-packet-only case isolates reconstruction of an unbounded attack:
after its smooth activation, the injected position grows linearly, while the
velocity and acceleration components are constant and zero, respectively.
Figure~\ref{fig:ramp-attack} therefore shows only the three state-packet
components. Their estimation errors remain bounded even though the position
attack itself is unbounded, consistently with the attack-rate dependence in
Theorem~\ref{thm:estimation}.

\begin{figure}[!t]
\centering
\includegraphics[width=0.96\columnwidth]{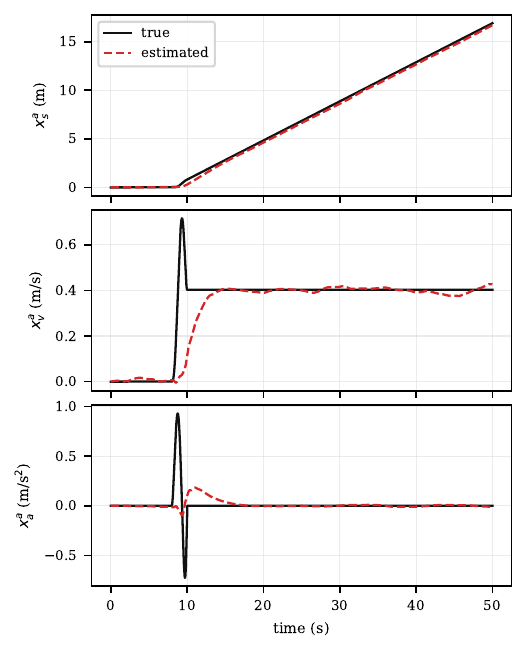}
\caption{Ramp all-link state-packet attack: injected and estimated signals on
representative link $(4,3)$.}
\label{fig:ramp-attack}
\end{figure}

For both time-varying attacks, the attack derivative, rather than its
amplitude, determines the persistent reconstruction error. Reporting each pair
as resilient/uncompensated, the final $\max_i\norm{e_i}$ is
$5.00/140.18$ for the low-frequency attack and $1.56/105.45$ for the ramp
attack, while the minimum intervehicle distance is $25.9$~m in both resilient
cases, against $3.6$~m and $13.3$~m without compensation. Thus compensation
keeps the intervehicle distances positive and reduces the final tracking error
by one to two orders of magnitude relative to the uncompensated controller.
The complete results, including the last-$5$~s RMS estimation errors over the
seven links, are provided in the public repository. These results demonstrate
resilience even when every controller-relevant predecessor link is
compromised, without an honest-majority condition.

\subsection{A single attacked link and downstream propagation}
To isolate downstream propagation, the low-frequency attack
\eqref{eq:sim-low-frequency} is next applied only to middle link $(4,3)$; all
other links remain attack-free. For each controller, the attack-induced
acceleration is
\[
 \Delta a_i(t):=a_i^{\rm attack}(t)-a_i^{\rm no\ attack}(t),
\]
where the reference trajectory uses the same switching signals, initial
condition, and communication-noise realization. Because the simulator is
linear under a fixed switching schedule, this subtraction isolates the
response caused by the attack without removing any attack-dependent dynamics.

Figure~\ref{fig:single-link-propagation} reports vehicles $4$--$7$. Over
$30\le t\le160$~s, the resilient normalized RMS values are
$[1,0.450,0.331,0.240]$, and the largest adjacent ratio is $0.736<1$.
Without compensation, the values are $[1,0.996,0.993,0.992]$, and the largest
adjacent ratio is $0.998$. The final largest tracking norms are $0.778$ and
$21.814$, respectively. Thus, the proposed method reduces the disturbance at
the directly affected follower and attenuates its downstream propagation.
This is a finite-horizon numerical observation under a nonzero estimation
error, rather than a claim that the strict nominal definition in
Definition~\ref{def:nominal-string} holds for every persistent attack.

\begin{figure}[t]
\centering
\includegraphics[width=\columnwidth]{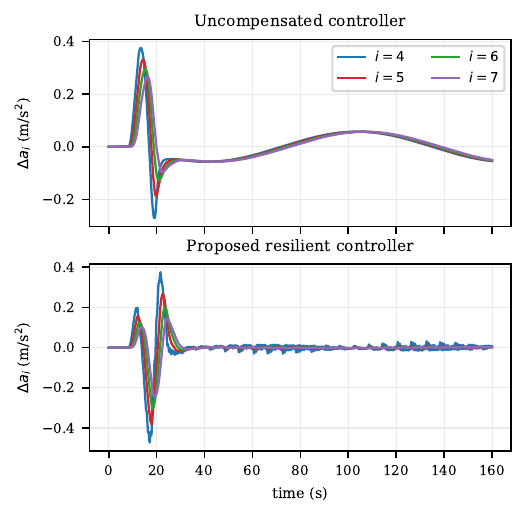}
\caption{Attack-induced acceleration under a low-frequency attack applied
only to link $(4,3)$. The bottom panels normalize by the RMS acceleration of
the directly affected vehicle~$4$.}
\label{fig:single-link-propagation}
\end{figure}

\subsection{Comparison with redundancy-based rejection}
The proposed per-link reconstruction is compared with the Mean Subsequence Reduced (MSR) platoon
controller of Zhao \emph{et al.}~\cite{Safeguard24}. To isolate the
attack-handling mechanism, the controlled comparison uses the same
seven-vehicle third-order plant, the fixed physical mode $p=2$, with
$\tau_2=0.5$~s and $k_2=[0.5,1,0.5]\T$, and the same three-predecessor
look-ahead graph for both methods; communication noise is omitted to isolate
attack cardinality. Thus, this comparison is not intended to assess performance
under powertrain switching; the switched-mode performance of the proposed
method is evaluated separately in the preceding experiments. The MSR controller removes the
neighbor-deviation vector farthest from the origin and averages the retained
vectors, whereas the proposed method estimates and compensates each link
separately. Every V2V link in the look-ahead graph is subjected to the same
smoothly activated constant attack. This violates the honest-neighbor
requirement of MSR filtering but remains admissible under
Assumption~\ref{as:attack}.

As shown in Fig.~\ref{fig:comparison}, the MSR rule has no honest packet to
retain once every incident link is compromised. Its steady largest tracking
error is $83.0$, and the minimum spacing falls to $9.4$~m. The proposed
method's corresponding values are $0.313$ and $25.83$~m. The purpose of this
comparison is not to claim that MSR fails within its stated threat model; it
shows the structural difference between an honest-majority method and the
present per-link reconstruction when all available links are attacked.

\begin{figure}[!t]
\centering
\includegraphics[width=\columnwidth]{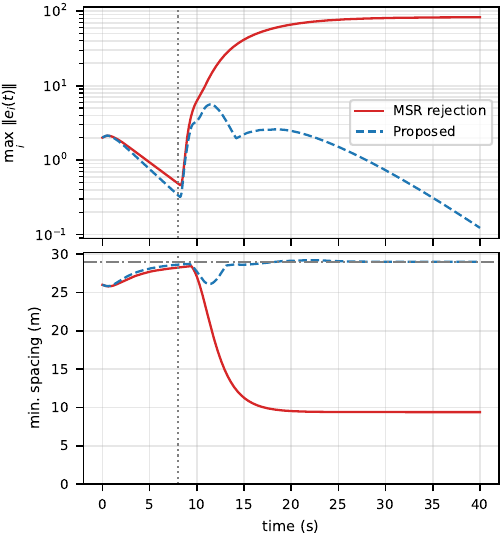}
\caption{All links attacked: MSR rejection~\cite{Safeguard24} versus the
proposed per-link reconstruction on the same redundant look-ahead graph.}
\label{fig:comparison}
\end{figure}
\vspace{-5pt}
\section{Conclusion}\label{sec:conclusion}
This paper developed a resilient CTH platoon architecture for switched
third-order vehicles under state-packet and auxiliary-output FDI attacks. A
defender-scheduled switching auxiliary system makes the augmented per-link system
uniformly observable, and the observer yields an estimation bound
that depends on attack rates rather than amplitudes. The resulting compensated
platoon is UUB under an explicit physical-switching condition, while the
predecessor-chain analysis separates nominal length-uniform acceleration
non-amplification from acceleration propagation under nonzero estimation
errors. The
exact frozen-mode test also isolates the headway mechanism that is absent from
constant-offset following. Numerical results with low-power communication
noise show bounded tracking under low-frequency and unbounded ramp attacks,
resilience when all
predecessor links are compromised, and downstream attenuation when only one
link is attacked.
The analysis adopts a common physical switching signal across the fleet, which
models a shared driving regime rather than vehicle-specific gear or actuation
changes; heterogeneous per-vehicle modes would turn the stacked closed loop
\eqref{eq:stacked-loop} into a genuinely time-varying interconnection and call
for a mode-mismatch string argument.

Future work will extend the analysis to asynchronous vehicle-specific
switching, communication delays and packet losses, actuator saturation, and
heterogeneous vehicle parameters. Further directions include stochastic-noise
robustness, invariant-set or control-barrier-function-based collision-avoidance
guarantees, and validation using higher-fidelity vehicle models and hardware experiments.
\bibliographystyle{IEEEtran}
\bibliography{references}
\end{document}